\documentclass[11pt]{article}
\usepackage{amsthm,amsfonts,amsmath}
\usepackage{enumerate,paralist}
\usepackage{xspace}
\usepackage{boxedminipage}
\usepackage{color}
\usepackage{hyperref}
\usepackage{cite}
\usepackage{graphicx, wrapfig}

\fontencoding{T1}

\advance \topmargin by -1.0in
\advance \oddsidemargin by -0.8in
\newcommand{\myparskip}{3pt}
\parskip \myparskip

\newcommand{\headers}[3]{
\newpage\setcounter{page}{1}
\def\@oddhead{$\underline{\hbox to\textwidth{%
\textbf{\rlap{#1}\phantom{hj}\hfill #2 \hfill \llap{#3}}}}$}
\def\@oddfoot{\hfill\thepage\hfill}}

\newtheorem{lemma}{Lemma}[section]
\newtheorem{theorem}[lemma]{Theorem}

\newtheorem{corollary}[lemma]{Corollary}
\newtheorem{prop}[lemma]{Proposition}

\newtheorem{remark}[lemma]{Remark}
\newtheorem{prob}{Problem}

\renewenvironment{proof}{\vspace{-0.1in}\noindent{\bf Proof:}}%
        {\hspace*{\fill}$\Box$\par}
        {\hspace*{\fill}$\Box$\par}
        {\hspace*{\fill}$\Box$\par}

\def\etal{\emph{et al.}\xspace}

\def\opt{\mathrm{OPT}}

\def\script#1{\mathcal{#1}}

\def\card#1{\left|#1\right|}
\def\set#1{\left\{#1\right\}}
\def\pair#1{\left<#1\right>}
\def\sep{\;|\;}

\def\sA{\script{A}}

\def\sG{\script{G}}

\def\sS{\script{S}}

\def\mypar#1{{\medskip\noindent \textbf{#1}}}

\def\prob#1{\textsf{\textup{#1}}\xspace}

\def\capacity{\mathrm{cap}}
\def\minDSlp{\prob{\minDS-LP}}
\def\minCDSlp{\prob{\minCDS-LP}}
\def\nwST{\prob{NW-Steiner-Tree}}
\def\nwSTlp{\prob{\nwST-LP}}

\def\NP{\mathrm{NP}}
\def\DTIME{\mathrm{DTIME}}
\def\PTAS{\mathrm{PTAS}}

\def\CDSpack{\prob{CDS-Packing}}

\def\capCDSpack{\prob{Cap-CDS-Packing}}

\def\minCDS{\prob{Min-Cost-CDS}}
\def\minDS{\prob{Min-Cost-DS}}

\def\cost{\mathrm{cost}}
\def\vx{\mathrm{\mathbf{x}}}
\def\vy{\mathrm{\mathbf{y}}}

\begin{document}

\title{Connected Domatic Packings in Node-capacitated Graphs}
\author{
Alina Ene\thanks{Dept.\ of Computer Science, University of Illinois,
Urbana, IL, 61801. Supported in part by NSF grant CCF-1016684 and a
Chirag Foundation graduate fellowship. This work was done while the
author was visiting the IBM Almaden Research center.
{\tt ene1@illinois.edu}.}
\and
Nitish Korula\thanks{
Google Research, New York, NY 10011. Part of this work was done while
the author was a student at University of Illinois.
{\tt nitish@google.com}.}
\and
Ali Vakilian\thanks{
Dept.\ of Computer Science, University of Illinois, Urbana, IL,
61801. Supported in part by NSF grant CCF-1016684 and a Siebel Scholar award. 
{\tt vakilia2@illinois.edu}.}}

\date{\today}

\pagenumbering{gobble}
\maketitle

\begin{abstract}
	A set of vertices in a graph is a dominating set if every vertex
	outside the set has a neighbor in the set. A dominating set is
	connected if the subgraph induced by its vertices is connected.
	The connected domatic partition problem asks for a partition of
	the nodes into connected dominating sets. The connected domatic
	number of a graph is the size of a largest connected domatic
	partition and it is a well-studied graph parameter with
	applications in the design of wireless networks.
	In this note, we consider the fractional counterpart of the
	connected domatic partition problem in \emph{node-capacitated}
	graphs. Let $n$ be the number of nodes in the graph and let $k$
	be the minimum capacity of a node separator in $G$. Fractionally
	we can pack at most $k$ connected dominating sets subject to the
	capacities on the nodes, and our algorithms construct packings
	whose sizes are proportional to $k$. Some of our main
	contributions are the following:
	\begin{itemize}
		\item An algorithm for constructing a fractional connected
		domatic packing of size $\Omega\left(k \right)$ for
		node-capacitated planar and minor-closed families of graphs.
		\item An algorithm for constructing a fractional connected
		domatic packing of size $\Omega\left(k / \ln{n} \right)$ for
		node-capacitated general graphs.
	\end{itemize}
\end{abstract}

\newpage
\setcounter{page}{1}
\pagenumbering{arabic}

\section{Introduction}
\label{sec:intro}

Let $G = (V, E)$ be an undirected and connected graph with $n$ nodes.
A set $S$ of nodes is a dominating set if every node not in $S$ has a
neighbor in $S$. A connected domatic partition is a collection of
connected dominating sets that are node disjoint. The connected
domatic number is the size of a largest connected domatic partition.
In this note, we consider the problem of constructing large
fractional connected domatic packings; a fractional packing is a
weight function on connected dominating sets such that, for each
vertex $v$, the total weight of the connected dominating sets that
contain $v$ is at most one.

Connected domatic partitions and packings have several applications
in the design of wireless networks. In these applications, a
connected dominating set is used as a virtual backbone, and the rest
of the nodes use the connected dominating set to exchange messages
and route traffic \cite{DasB97,DasSB97,MahjoubM10}.  Motivated by the
goal of improving the energy efficiency and the lifetime of the
network, several papers \cite{MisraM09,MoscibrodaW05,PemmarajuP06}
have proposed using several connected dominating sets; these
approaches first compute a large connected domatic packing or
partition and they rotate between the connected dominating sets.
Additionally, the recent work of Censor-Hillel \etal \cite{CHGK}
establishes a close connection between the fractional connected
domatic number and the throughput of store-and-forward algorithms for
routing in wireless networks.

Integer and fractional packings of combinatorial structures are
connected to each other and to the corresponding optimization problem
that asks for the minimum cost combinatorial structure; we refer the
reader to Section~{5} in \cite{CalinescuCV09} for an overview of
these connections. In particular, an $\alpha$-approximation for the
minimum-cost \prob{Connected Dominating Set} (\minCDS) problem implies an
$\alpha$-approximation for the \prob{Connected Domatic Packing}
(\CDSpack) problem; this connection was shown by
Carr and Vempala \cite{CarrV02}.  This result and the $O(\ln{n})$
approximation algorithm for \minCDS given by Guha and Khuller
\cite{GuhaK98} imply an $O(\ln{n})$ approximation for \CDSpack. In
very recent work, Censor-Hillel \etal \cite{CHGK} gave the first
poly-logarithmic approximation for the \prob{Connected Domatic
Partition} problem; their algorithm achieves
an $O(\ln^5{n})$ approximation. The results of \cite{CHGK} guarantee
partitions and packings whose sizes are a poly-logarithmic fraction
of the \emph{vertex connectivity}\footnote{A graph $G = (V, E)$ is
$k$-vertex-connected iff, for any subset $S \subseteq V $ of size
less than $k$, the removal of $S$ does not disconnect the graph. The
vertex connectivity of $G$ is the maximum $k$ such that $G$ is
$k$-vertex-connected.}. These guarantees are independent of the size
of the largest partition or packing and thus they are not
approximation results \emph{per se}.  Since the connectivity of the
graph is an upper bound on the fractional connected domatic number
and thus the connected domatic number as well, these absolute results
give us approximation guarantees as a byproduct.

In several applications in wireless networks, each node has a certain
battery life that constrains how long the node can be used as part of
a virtual backbone for the network. We can model such networks using
node-capacitated graphs, where the capacity represents the battery
life of the node. Motivated in part by these applications, we
consider the more general problem of constructing large connected
packings in node-capacitated graphs. In this setting, each vertex $v$
has a capacity $\capacity(v)$ and the goal is to find a fractional
packing of maximum total weight such that the fractional weight of
the connected dominating sets that contain each vertex is at most the
capacity of the vertex. We refer to the capacitated analogue of
\CDSpack as \capCDSpack.  We can reduce the capacitated problem to
the uncapacitated one by replacing each node by a clique whose size
is equal to the capacity of the node. However, this reduction does
not run in polynomial time if the capacities are large and it does
not preserve the special structure of certain graphs, such as planar
or minor-free graphs.  Real-world wireless networks are typically not
arbitrary graphs but rather they are nearly planar or have restricted
structure. We give an algorithm that constructs improved fractional
packings for such networks.

\begin{theorem} \label{thm:cds-packing-minor-free}
	Let $G$ be a node-capacitated graph that belongs to a
	minor-closed family $\sG$ of graphs. Let $k$ be the minimum
	capacity of a node separator\footnote{A set $S$ is a node
	separator in $G$ if the graph $G-S$ has at least two connected
	components, where $G-S$ is the graph obtained from $G$ by
	removing the nodes of $S$} in $G$. There is a polynomial time
	algorithm that constructs a fractional connected domatic packing
	in $G$ of size $\Omega(k)$, where the constant depends only on
	the family $\sG$.
\end{theorem}

\noindent
Our approach can also be used to construct fractional packings for
general graphs with arbitrary node capacities. This result was
shown in \cite{CHGK} for uncapacitated graphs using very different
techniques.

\begin{theorem} \label{thm:cds-packing-general}
	Let $G$ be a node-capacitated graph. Let $k$ be the minimum
	capacity of a node separator in $G$. There is a polynomial time
	algorithm that constructs a fractional connected domatic packing
	in $G$ of size $\Omega(k / \ln{n})$.
\end{theorem}

\noindent
Our algorithm for \capCDSpack is based on a connection between the
size of a fractional packing and the integrality gap of a standard LP
relaxation for \minCDS; we describe this LP relaxation in
Section~\ref{sec:cds-packing}. We show that, if the relaxation has an
integrality gap of $r$, we can construct a packing of size $k/r$ in
polynomial time using an $r$-approximate rounding algorithm for the
\minCDS LP and the ellipsoid method. One of our contributions is a
constant upper bound on the integrality gap of \minCDS LP in
minor-closed families of graphs, where the constant depends only on
the family.  In the process, we also show that the integrality gap of
a standard LP relaxation for the minimum cost \prob{Dominating Set}
(\minDS) problem is constant in minor-free graphs. The \minDS problem
admits a $\PTAS$ in planar graphs \cite{Baker94}, but this result
does not establish an upper bound on the integrality gap. Our
algorithms can be easily adapted to give analogous integrality gap
upper bounds for the Steiner variants of \minDS and \minCDS; in the
Steiner problems, we are given a subset of the vertices called the
terminals and the goal is to select a (connected) set that dominates
the terminals.

\begin{theorem} \label{thm:ds-gap-minor-free}
	The standard LP relaxation for the \minDS problem has an $O(1)$
	integrality gap in planar and minor-closed families of graphs.
	Moreover, there is a polynomial time algorithm that rounds any
	fractional solution to an integral solution whose cost is at most
	$O(1)$ times larger than the cost of the fractional solution.
\end{theorem}

\begin{theorem} \label{thm:cds-gap-minor-free}
	The standard LP relaxation for the \minCDS problem has an $O(1)$
	integrality gap in planar and minor-closed families of graphs.
	Moreover, there is a polynomial time algorithm that rounds any
	fractional solution to an integral solution whose cost is at most
	$O(1)$ times larger than the cost of the fractional solution.
\end{theorem}

\vspace{-0.1in}
\mypar{Other related work:}
Domatic partitions have received considerable attention; we refer the
reader to \cite{HedetniemiL91, HaynesHS98a, HaynesHS98b} for a
comprehensive treatment of graph domination. Feige \etal
\cite{FeigeHKS02} gave a polynomial time algorithm that constructs a
domatic partition of size $\Omega(\delta/\ln{n})$, where $\delta$ is
the minimum degree of the graph and they showed that this is best
possible unless $\NP \subseteq \DTIME\left(n^{\log\log{n}}\right)$.
C\u{a}linescu \etal \cite{CalinescuCV09} considered the more general
problem of packing disjoint bases in a polymatroid.

\section{Algorithm for fractional connected domatic packings}
\label{sec:cds-packing}

In this section, we give polynomial time algorithms for constructing
fractional connected domatic packings in node-capacitated graphs.

We start by introducing the following natural LP relaxation for the
\minCDS problem. Let $G = (V, E)$ be a graph with costs $\cost(v)$
associated with the nodes. For each vertex $v$, we let $\Gamma(v)$
denote the set of all neighbors of $v$ in $G$. For a set $S$ of
nodes, we let $\Gamma(S)$ denote the set of all nodes $v$ such that
$v$ is not in $S$ and $v$ has a neighbor in $S$.  Let $\Gamma^+(v) =
\Gamma(v) \cup \set{v}$. The relaxation has a variable $x(v)$ for
each vertex $v$ with the interpretation that $x(v) = 1$ iff $v$ is in
the connected dominating set. Let $\sS$ be the collection of all sets
$S$ such that $S$, $\Gamma(S)$, and $V - (S \cup \Gamma(S))$ are all
non-empty; note that, for each set $S \in \sS$, the set $\Gamma(S)$
is a node separator that separates $S$ from $V - (S \cup \Gamma(S))$.
The relaxation \minCDSlp is given below.

\begin{center}
\begin{boxedminipage}{0.45\textwidth}
\vspace{-0.15in}
\begin{align*}
	&  \minCDSlp &\\
	\quad \min \quad & \sum_{v \in V} x(v) \cost(v) &\\
	\text{s.t.} \quad & \sum_{ u \in \Gamma^+(v)} x(u) \geq 1 \quad &
	v \in V\\
	& \sum_{v \in \Gamma(S)} x(v) \geq 1 \qquad\quad & S \in \sS\\
	& x(v) \geq 0 & v\in V
\end{align*}
\end{boxedminipage}
\end{center}

Note that the LP is a valid relaxation for the \minCDS problem. A
dominating set must contain a vertex from $\Gamma^+(v)$ for each
vertex $v$. Additionally, a connected dominating set must contain a
vertex from each node separator.

The two main steps of our approach for constructing large fractional
packings are the following. The first step is to show that we can
construct in polynomial time a packing of size $\Omega(k/r)$, where
$k$ is the capacity of a minimum node separator in $G$ and $r$ is an
upper bound on the integrality gap of \minCDSlp; we refer the reader
to Corollary~\ref{cor:cds-pack-decomp} for a precise statement of the
result.  The second step is to upper bound the integrality gap of
\minCDSlp. For general graphs, it follows easily from previous work
that the integrality gap is $O(\ln{n})$ and thus we can find a
packing of size $\Omega(k/\ln{n})$. For planar graphs and more
generally, minor-free graphs, we will show that the integrality gap
of \minCDSlp is a constant and thus we can find a packing of size
$\Omega(k)$.

In the following, we say that a rounding algorithm $\sA$ for an LP
relaxation is an \emph{$r$-approximate} rounding algorithm for the
LP if, given any fractional solution to the LP, the algorithm
constructs an integral solution of value at most $r$ times the value
of the fractional solution.

Consider an instance $\pair{G, \capacity}$ of \capCDSpack, where
$G$ is a graph from a family $\sG$ of graphs and $\capacity(\cdotp)$
is a capacity function on the nodes of $G$. Let $k$ be the minimum
capacity of a node separator in $G$. Our goal is to show that we can
construct a fractional packing of size $\Omega(k/r)$ provided that we
have an $r$-approximate rounding algorithm for \minCDSlp.  This will
follow from the theorem below, which is an immediate corollary of
Theorem~{2} in \cite{CarrV02}.

\begin{theorem}[Carr and Vempala \cite{CarrV02}] \label{thm:convex-decomp}
	Let $\sG$ be a family of graphs. Let $\vx$ be a fractional
	solution to \minCDSlp for an instance of \minCDS for which the
	graph $G$ is in $\sG$. Let $\sA$ be a polynomial time rounding
	algorithm for \minCDSlp that is $r$-approximate on instances for
	which the graph is in $\sG$. Given $\vx$ and $\sA$, we can find
	in polynomial time a collection of polynomially many connected
	dominating sets $D_1, \dots, D_{\ell}$ with associated weights
	$\lambda_1, \dots, \lambda_{\ell}$ such that $\sum_{i = 1}^{\ell}
	\lambda_i = 1$ and, for each vertex $v$, we have $\sum_{i: v \in
	D_i} \lambda_i \leq r \cdotp x(v)$.
\end{theorem}

\noindent
The theorem above gives us the following corollary.

\begin{corollary} \label{cor:cds-pack-decomp}
	Let $\sG$ be a family of graphs. Let $\sA$ be a polynomial time
	rounding algorithm for \minCDSlp that is $r$-approximate on
	instances for which the graph is in $\sG$.  Let $\pair{G,
	\capacity}$ be an instance of \capCDSpack such that $G \in \sG$.
	Let $k$ be the minimum capacity of any node separator in $G$.
	Given $\sA$ and $\pair{G, \capacity}$, we can find in polynomial
	time a collection of polynomially many connected dominating sets
	$D_1, \dots, D_{\ell}$ and associated weights $\alpha_1, \dots,
	\alpha_{\ell}$ such that $\sum_{i = 1}^{\ell} \alpha_i \geq k/r$
	and, for each vertex $v \in V(G)$, we have $\sum_{i: v \in D_i}
	\alpha_i \leq \capacity(v)$. Differently said, $\set{\pair{D_1,
	\alpha_1}, \dots, \pair{D_{\ell}, \alpha_{\ell}}}$ is a feasible
	fractional connected domatic packing of size $\Omega(k/r)$.
\end{corollary}
\begin{proof}
	Consider the following fractional solution $\vx$: $x(v) =
	\capacity(v)/k$ for each vertex $v \in V(G)$. We can verify that
	$\vx$ is a feasible solution to \minCDSlp as follows. Consider a
	vertex $v$. We can assume that $\Gamma(v)$ is a node separator;
	otherwise, $\Gamma^+(v)= V(G)$ and $\sum_{u\in
	\Gamma^+(v)}x(u)\geq 1$ trivially holds. Since $\Gamma(v)$ is a
	node separator, it follows that $\capacity(\Gamma(v)) \geq k$.
	Therefore we have
		$$\sum_{u \in \Gamma(v)} x(u) = {1 \over k} \sum_{u \in
		\Gamma(v)} \capacity(u) \geq 1,$$
	and thus $\vx$ satisfies the first set of constraints. Consider a
	set $S \in \sS$. Since $\Gamma(S)$ is a node separator in $G$ it
	follows that $\capacity(\Gamma(S)) \geq k$. Therefore we have
		$$\sum_{v \in \Gamma(S)} x(v) = {1 \over k} \sum_{v \in
		\Gamma(S)} \capacity(v) \geq 1,$$
	and thus $\vx$ satisfies the second set of constraints.

	We apply Theorem~\ref{thm:convex-decomp} to $\vx$ and $\sA$ in
	order to get a collection of connected dominating sets $D_1,
	\dots, D_{\ell}$ and associated weights $\lambda_1, \dots,
	\lambda_{\ell}$. For each $i$, let $\alpha_i = (k/r) \cdotp
	\lambda_i$. We can verify that $\set{\pair{D_1, \alpha_1}, \dots,
	\pair{D_{\ell}, \alpha_{\ell}}}$ is the desired packing as
	follows. We have
		$$\sum_{i = 1}^{\ell} \alpha_i = {k \over r} \sum_{i =
		1}^{\ell} \lambda_i = {k \over r}.$$
	Additionally, for each vertex $v$, we have
		$$\sum_{i: v \in D_i} \alpha_i = {k \over r} \sum_{i: v \in
		D_i} \lambda_i \leq {k \over r} \cdotp r \cdotp x(v) =
		\capacity(v).$$
\end{proof}

\medskip\noindent
In the second step, we upper bound the integrality gap of \minCDSlp.
To this end, we will first relate the integrality gap of \minCDSlp to
the integrality gaps of the standard LP relaxations for the
minimum-cost \prob{Dominating Set} (\minDS) problem and the minimum
node-weighted \prob{Steiner Tree} (\nwST) problem, and then we will
upper bound the integrality gaps of these two relaxations. The LP
relaxation for \minDS is the relaxation \minDSlp given below.  In the
\nwST problem, we are given a graph $G = (V, E)$ with non-negative
weights $w(v)$ on the nodes and a set $T \subseteq V$ of nodes called
\emph{terminals}. The goal is to select a minimum weight subgraph $H$
of $G$ that spans all the terminals, where the weight of $H$ is the
total weight of the nodes in $H$. (Note that we may assume that $H$
is a node-induced connected subgraph of $G$.) Let $\sS_T$ be the
collection consisting of all sets $S$ such that $S$ separates the
terminals; more precisely, $S \cap T$ and $(V - S) \cap T$ are both
non-empty.  For each set $S \in \sS_T$, at least one vertex in
$\Gamma(S)$ must be in the solution. The LP relaxation for \nwST is
the relaxation \nwSTlp given below.

\begin{center}
\begin{boxedminipage}{0.45\textwidth}
\vspace{-0.15in}
\begin{align*}
	&  \minDSlp &\\
	\quad \min \quad & \sum_{v \in V} x(v) \cost(v) &\\
	\text{s.t.} \quad & \sum_{ u \in \Gamma^+(v)} x(u) \geq 1 \quad &
	v \in V\\
	& x(v) \geq 0 & v\in V
\end{align*}
\end{boxedminipage}
\begin{boxedminipage}{0.45\textwidth}
\vspace{-0.15in}
\begin{align*}
	&  \nwSTlp &\\
	\quad \min \quad & \sum_{v \in V} x(v) w(v) &\\
	\text{s.t.} \quad & \sum_{v \in \Gamma(S)} x(v) \geq 1 \quad &
	S \in \sS_T\\
	& x(v) \geq 0 & v\in V
\end{align*}
\end{boxedminipage}
\end{center}

\noindent
The following straightforward propositions allow us to relate the
integrality gap of \minCDSlp to the integrality gaps of \minDSlp and
\nwSTlp.

\begin{prop} \label{prop:feasible-dominating-set}
	Consider an instance of \minCDS; let $G$ be the input graph and
	let $\vx$ be a feasible solution to \minCDSlp for this instance.
	Then $\vx$ is a feasible solution to \minDSlp for any instance of
	\minDS in which the input graph is $G$.
\end{prop}

\begin{prop} \label{prop:feasible-steiner-tree}
	Consider an instance of \minCDS; let $G$ be the input graph and
	let $\vx$ be a feasible solution to \minCDSlp for this instance.
	Let $T$ be any subset of the vertices and let $\vx'$ be the
	following fractional solution: $x'(v) = x(v)$ if $v \notin T$ and
	$x'(v) = 1$ otherwise. Then $\vx'$ is a feasible solution to
	\nwSTlp for any instance of \nwST in which the input graph is $G$
	and the set of terminals is $T$.
\end{prop}
\begin{proof}
	Consider a set $S \in \sS_T$. If $V - (S \cup \Gamma(S))$ is
	non-empty, $S \in \sS$ and the fact that $\vx$ is a feasible
	solution to \minCDSlp gives us that $x(\Gamma(S))$ is at least
	one. Therefore we may assume that $S \cup \Gamma(S) = V$. Since
	$S$ separates the terminals, $\Gamma(S)$ contains a terminal and
	thus $x'(\Gamma(S))$ is at least one.
\end{proof}

\begin{corollary}
	Let $\sG$ be a family of graphs.  Let $\sA_1$ be a polynomial
	time rounding algorithm for \minDSlp that is $r_1$-approximate
	on instances in which the graph is in $\sG$. Let
	$\sA_2$ be a polynomial time rounding algorithm for \nwSTlp that
	is $r_2$-approximate on instances in which the graph is in $\sG$.
	Given $\sA_1$ and $\sA_2$, we can design a
	polynomial time a rounding algorithm for \minCDSlp that is
	$(r_1 + r_2)$-approximate on instances in which the graph is in
	$\sG$.
\end{corollary}
\begin{proof}
	Let $\pair{G, \cost}$ be an instance of \minCDS, where $G \in
	\sG$. Let $\vx$ be a feasible solution to \minCDSlp for this
	instance. Let $C = \sum_{v \in V} x(v) \cost(v)$. Our goal is to
	show that we can construct in polynomial time a connected
	dominating set $D'$ whose cost $\cost(D')$ is at most
	$(r_1+r_2)C$. By Proposition~\ref{prop:feasible-dominating-set},
	$\vx$ is a feasible solution to \minDSlp for the instance
	$\pair{G, \cost}$.  Thus we can run $\sA_1$ with $\vx$ as input
	in order to get a dominating set $D$ such that $\cost(D) \leq r_1
	\cdotp C$. Once we have the dominating set $D$, we consider the
	following instance of \nwST. The nodes in $D$ will be the
	terminals. We define a set of weights as follows: for each vertex
	$v$, we have $w(v) = \cost(v)$ if $v \notin D$ and $w(v) = 0$
	otherwise. We define a fractional solution $\vx'$ as follows: for
	each vertex $v$, we have $\vx'(v) = \vx(v)$ if $v \notin D$ and
	$\vx'(v) = 1$ otherwise. Let $W = \sum_{v \in V} w(v) x'(v)$;
	note that $W = \sum_{v \in V-D} x(v) \cost(v) \leq C$. By
	Proposition~\ref{prop:feasible-steiner-tree}, $\vx'$ is a
	feasible solution to \nwSTlp for the instance $\pair{G, w, D}$.
	Thus we can run $\sA_2$ with $\vx'$ as input in order to get a
	node-induced connected subgraph $H$ of $G$ that spans $D$ and it
	has weight $w(H) \leq r_2 \cdotp W$. Let $D' = V(H)$; since $D$
	is a subset of $D'$, $D'$ is a dominating set. Additionally,
	$\cost(D') \leq (r_1 + r_2) C$.
\end{proof}

\medskip\noindent
Consider the relaxation \minDSlp. For general graphs, we can show an
$O(\ln{n})$ upper bound on the integrality gap using the following
standard randomized rounding approach. Given a fractional
solution $\vx$, we select a set $D$ of nodes as follows: for each
vertex $v$, we add $v$ to $D$ independently at random with
probability $\min\set{c\ln{n} \cdotp x(v), 1}$, where $c$ is a large
enough constant. With high probability, the resulting set $D$ is a
dominating set. For minor-closed families of graphs, we give a
primal-dual algorithm in Section~\ref{sec:min-cost-ds-planar} that
shows that the integrality gap is $O(1)$. We remark that the \minDS
problem admits a $\PTAS$ in planar graphs \cite{Baker94}, but the
algorithm of \cite{Baker94} does not give an upper bound on the
integrality gap of the LP.

\begin{theorem} \label{thm:min-ds-planar-integrality-gap}
	Let $\sG$ be a minor-closed family of graphs. There is a
	polynomial time rounding algorithm for \minDSlp that is
	$c(\sG)$-approximate on instances in which the graph is in $\sG$,
	where $c(\sG)$ is a constant that depends only on the family
	$\sG$.
\end{theorem}

\noindent
Finally, consider the relaxation \nwST. Guha \etal \cite{GuhaMNS99}
showed that the integrality gap is $O(\ln{n})$ for general graphs,
and Demaine \etal \cite{DemaineHK09} showed that the integrality gap
is $O(1)$ for minor-closed families of graphs.  This completes the
proof of Theorem~\ref{thm:cds-packing-general} and
Theorem~\ref{thm:cds-packing-minor-free}.

\section{Algorithm for \minDS in minor-closed families of graphs}
\label{sec:min-cost-ds-planar}

In this section, we give a primal-dual algorithm for the minimum cost
\prob{Dominating Set} problem (\minDS) in minor-closed families of
graphs that achieves a constant factor approximation. The algorithm
will also establish a matching upper bound on the integrality gap of
the standard LP relaxation for the problem that was given in
Section~\ref{sec:cds-packing}.

Let $G = (V, E)$ be a node-weighted graph, and let $\cost(v)$ denote
the cost of $v$. As before, for each vertex $v$, we let $\Gamma(v)$
denote the set of all neighbors of $v$ in $G$. Let $\Gamma^+(v) =
\Gamma(v) \cup \set{v}$. The primal and dual LPs are described below;
we omit the constraint $x(v) \leq 1$ from the primal LP, since it is
redundant.

\begin{center}
\begin{boxedminipage}{0.45\textwidth}
\vspace{-0.15in}
\begin{align*}
	&  \minDSlp &\\
	\quad \min \quad & \sum_{v \in V} x(v) \cost(v) &\\
	\text{s.t.} \quad & \sum_{ u \in \Gamma^+(v)} x(u) \geq 1 \quad &
	v\in V\\
	& x(v) \geq 0 & v\in V
\end{align*}
\end{boxedminipage}
\vspace{0.1in}
\begin{boxedminipage}{0.45\textwidth}
\vspace{-0.15in}
\begin{align*}
	&  \text{Dual of } \minDSlp &\\
	\quad \max \quad & \sum_{v \in V} y(v)\\
	\text{s.t.} \quad & \sum_{u \in \Gamma^+(v)} y(u) \leq
	\cost(v) & v\in V\\
	& y(v) \geq 0 & v\in V
\end{align*}
\end{boxedminipage}
\end{center}

\noindent
The algorithm is based on the primal-dual framework of Goemans and
Williamson \cite{GoemansW95}. The algorithm selects a dominating set
$X$ for $G$. Initially, $X$ consists of all vertices with zero cost.
We also maintain a dual solution $\vy$; initially, $y(v) = 0$ for all
$v \in V$. We proceed in iterations. Consider an iteration $i$ and
let $X_{i - 1}$ be the set of nodes selected in the first $i - 1$
iterations.  Let $A_i$ be the set of all vertices $v \in V$ such that
$X_{i - 1} \cap \Gamma^+(v)$ is empty. If $A_i$ is empty, $X_{i - 1}$
is a dominating set and we return $X_{i - 1}$.  Otherwise, we
increase the dual variables $\set{y(a) \sep a \in A_i}$ uniformly
until a dual constraint for a node $v$ becomes tight, i.e., we have
$\sum_{u \in \Gamma^+(v)} y(v) = \cost(v)$; we add all the tight
vertices to $X$.

We note that, in each iteration $i$, it is possible to increase the
dual variables corresponding to the nodes of $A_i$. The set $X_{i -
1}$ contains all the vertices whose dual constraints are tight at the
beginning of iteration $i$. Thus, at the beginning of iteration $i$,
for each vertex $a \in A_i$ and each vertex $v$ such that $a \in
\Gamma^+(v)$, the dual constraint corresponding to $v$ is slack,
i.e., we have $\sum_{u \in \Gamma^+(v)} y(u) < \cost(v)$. Therefore
the algorithm terminates in at most $n$ iterations.

Finally, we perform a \emph{reverse-delete step}. Let $X$ be the
dominating set selected by the primal-dual algorithm. We select a
subset $Y$ as follows. We start with $Y = X$. We order the vertices
of $Y$ in the reverse of the order in which they were selected by the
primal-dual algorithm. We consider the vertices of $Y$ in this order.
Let $v$ be the current vertex. If $Y - v$ is a dominating set, we
remove $v$ from $Y$.

The algorithm described above is well-defined on general graphs, but
its approximation is $\Omega(n)$. In the following, we show that we
can take advantage of the fact that minor-free graphs are sparse in
order to show that the algorithm achieves a constant factor
approximation in minor-closed families of graphs; the constant
depends on the family.

We start by noting that the dual solution $\vy$ satisfies the
complementary slackness conditions.

\begin{prop} \label{prop:complementary-slackness}
	For each vertex $v \in Y$, we have $\sum_{u \in \Gamma^+(v)}
	y(u) = \cost(v)$.
\end{prop}

\noindent
The following lemma gives us a very convenient way to upper bound the
approximation ratio. The lemma follows from a standard primal-dual
analysis and the fact that the algorithm increases the dual variables
uniformly in each iteration. Recall that $Y$ is the final dominating
set after performing reverse-delete, and $X_{i - 1}$ is the set of
vertices selected in the first $i - 1$ iterations of the algorithm.

\begin{lemma} \label{lem:approx-condition}
	Let $W_i = Y - X_{i - 1}$. Suppose that there exists a $\gamma$
	such that, for each iteration $i$ of the algorithm, we have
		$$\sum_{v \in A_i} \card{W_i \cap \Gamma^+(v)} \leq \gamma
		\card{A_i}.$$
	Then the cost of $Y$ is at most $\gamma \cdotp \opt$, where
	$\opt$ is the cost of the optimal solution to \minDSlp.
\end{lemma}
\begin{proof}
	By Proposition~\ref{prop:complementary-slackness}, we have
		$$\sum_{v \in Y} \cost(v) = \sum_{v \in Y} \sum_{u \in
		\Gamma^+(v)} y(u).$$
	By rearranging the second summation, we get that
		$$\sum_{v \in Y} \cost(v) = \sum_{v \in Y} \sum_{u \in
		\Gamma^+(v)} y(u) = \sum_{v \in V} y(v) \card{Y \cap
		\Gamma^+(v)}.$$
	Since $\vy$ is a feasible dual solution, by weak duality, we have
		$$\opt \geq \sum_{v \in V} y(v).$$
	Therefore it suffices to show that
		$$\sum_{v \in V} y(v) \card{Y \cap \Gamma^+(v)} \leq
		\gamma \sum_{v \in V} y(v).$$
	We can prove the inequality above by induction on the number of
	iterations. Initially, $y(v) = 0$ for all vertices $v$ and the
	inequality clearly holds. Now consider an iteration $i \geq 1$.
	Let $\epsilon$ be the amount by which the dual variables
	$\set{y(a) \sep a \in A_i}$ are increased in iteration $i$. The
	right-hand side of the inequality increases by $\epsilon
	\card{A_i}$. Thus, if we can show that the left-hand side
	increases by at most $\epsilon \gamma \card{A_i}$, the inequality
	will follow. The left hand side of the inequality
	increases by $\epsilon \sum_{v \in A_i} \card{Y \cap
	\Gamma^+(v)}$. For each $v \in A_i$, we have $\Gamma^+(v) \cap
	X_{i - 1}$ is empty, and thus
		$$\sum_{v \in A_i} \card{Y \cap \Gamma^+(v)} = \sum_{v \in
		A_i} \card{W_i \cap \Gamma^+(v)} \leq \gamma \card{A_i},$$
	where the last inequality follows from the assumption in the
	statement of the lemma. Therefore the left-hand side increases by
	at most $\epsilon \gamma \card{A_i}$, and the lemma follows.
\end{proof}

\medskip\noindent
Therefore, in order to upper bound the approximation ratio of the
algorithm, it suffices to prove the following key lemma. The lemma
follows from the minimality of $Y$ and the fact that minor-free
graphs are sparse, in the sense that the number of edges is
proportional to the number of vertices.

\begin{lemma} \label{lem:counting}
	Suppose that the input graph $G$ belongs to a minor-closed
	familty $\sG$ of graphs. There is a constant $c({\sG})$ depending
	only on $\sG$ such that, for each iteration $i$ of the algorithm,
	we have
		$$\sum_{u \in A_i} \card{W_i \cap \Gamma^+(u)} \leq c({\sG})
		\cdotp \card{A_i},$$
	where $W_i = Y - X_{i - 1}$.
\end{lemma}

\noindent
We devote the rest of this section to the proof of
Lemma~\ref{lem:counting}. We will prove the lemma in two steps. In
the first step, we use the sparsity of minor-free graphs to show that
the sum $\sum_{u \in A_i} \card{W_i \cap \Gamma^+(u)}$ is
at most a constant times larger than $\card{A_i} + \card{W_i}$. In
the second step, we use the minimality of $W_i$ to show that
$\card{W_i} \leq \card{A_i}$.

\begin{lemma} \label{lem:counting-first}
	Let $c'(\sG)$ be a constant such that, for each graph $K \in
	\sG$, we have $\card{E(K)} \leq c'(\sG) \card{V(K)}$.  For each
	iteration $i$, we have
		$$\sum_{u \in A_i} \card{W_i \cap \Gamma^+(u)} \leq
		\card{A_i \cap W_i} + c'(\sG) (\card{A_i} + 3\card{W_i}).$$
\end{lemma}
\begin{proof}
	Consider an iteration $i$ of the algorithm.  We have
		$$\sum_{u \in A_i} \card{W_i \cap \Gamma^+(u)} = \card{A_i
		\cap W_i} + \sum_{u \in A_i} \card{W_i \cap \Gamma(u)}.$$
	We can upper bound the second sum in the equation above as
	follows.  Let $G_1$ be the subgraph of $G$ whose vertices are
	$A_i \cup W_i$ and whose edges are all the edges of $G$ with one
	endpoint in $A_i - W_i$ and the other in $W_i$. Note that
	$\sum_{u \in A_i - W_i} \card{W_i \cap \Gamma(u)}$ is equal to
	the number of edges of $G_1$.  Let $G_2$ be the subgraph of $G$
	whose vertices are $W_i$ and whose edges are all the edges of $G$
	with one endpoint in $A_i \cap W_i$ and the other in $W_i - A_i$.
	Finally, let $G_3 = G[A_i \cap W_i]$ be the subgraph of $G$
	induced by $A_i \cap W_i$. Note that $\sum_{u \in A_i \cap W_i}
	\card{W_i \cap \Gamma(u)}$ is equal to the number of edges of
	$G_2$ plus the number of edges of $G_3$. Therefore we have
		$$\sum_{u \in A_i} \card{W_i \cap \Gamma(u)} = \card{E(G_1)}
		+ 2\card{E(G_2)} + \card{E(G_3)}.$$
	Therefore we have
		$$\sum_{u \in A_i} \card{W_i \cap \Gamma(u)} \leq c'(\sG)
		(\card{V(G_1)} + 2\card{V(G_2)} + \card{V(G_3)}) = c'(\sG)
		(\card{A_i} + 3\card{W_i}).$$
\end{proof}

\begin{lemma} \label{lem:counting-second}
	For each iteration $i$, we have $\card{W_i} \leq \card{A_i}$.
\end{lemma}
\begin{proof}
	Consider a vertex $w \in W_i$. We claim that, since we could not
	remove $w$ in the reverse-delete step, there is a vertex $v \in
	A_i$ such that $\Gamma^+(v) \cap (X_{i - 1} \cup Y) = \set{w}$.
	We can show this as follows. Since we could not remove $w$, there
	is a vertex $v \in V - (Y \cup X_{i - 1})$ such that $\Gamma^+(v)
	\cap (Y \cup X_{i - 1}) = \set{w}$. Since $v$ is not dominated by
	$X_{i - 1}$, $v$ is in $A_i$. Thus each vertex $w \in W_i$ has a
	\emph{witness} vertex $v \in A_i$ such that $\Gamma^+(v) \cap
	(X_{i - 1} \cup Y) = \set{w}$. Now we claim that each vertex $v
	\in A_i$ is a witness vertex for at most one vertex of $W_i$.
	Suppose for contradiction that a vertex $v \in A_i$ is a witness
	vertex for two vertices $w_1$ and $w_2$ in $W_i$. Without loss of
	generality, $w_1$ was selected by the algorithm after $w_2$.
	Consider the iteration of the reverse-delete step that considered
	$w_1$. At this point $w_2$ had not been considered yet and thus
	it is in $Y$. Thus $w_2 \in \Gamma^+(v) \cap (X_{i - 1} \cup Y)$,
	which contradicts the fact that $\Gamma^+(v) \cap (X_{i - 1} \cup
	Y) = \set{w_1}$. Therefore $\card{W_i} \leq \card{A_i}$, as
	desired.
\end{proof}

\medskip\noindent
Lemma~\ref{lem:counting} follows from Lemma~\ref{lem:counting-first}
and Lemma~\ref{lem:counting-second}. We have the following upper
bounds on the constant $c'(\sG)$ (see
Lemma~\ref{lem:counting-first}). If $\sG$ is a minor-closed family,
there is a constant-sized graph $H$ such that $\sG$ is the family of
all graphs that do not have $H$ as a minor. As shown by Kostochka
\cite{Kostochka84}, we have $c'(\sG) = O(\sqrt{\log(\card{V(H)})})$ for
the family of $H$-minor-free graphs. If $G$ is a planar graph, we
have $c'(\sG) < 3$; if $G$ is also bipartite, the constant improves
to $2$. Thus the algorithm achieves an $10$-approximation for planar
graphs.

\begin{remark}
	The algorithm above can be easily adapted to give a constant
	factor approximation for the minimum cost \prob{Steiner
	Dominating Set} problem in minor-closed families of graphs. In
	the Steiner problem, we are given a subset of vertices called
	terminals and the goal is to select a set that dominates the
	terminals.
\end{remark}

\begin{remark}
	A constant factor approximation for the minimum cost \prob{Steiner
	Dominating Set} problem in minor-free graphs can also be obtained
	via iterated rounding.
\end{remark}

\mypar{Acknowledgements:}
The results in Section~\ref{sec:cds-packing} were developed in joint
work with Chandra Chekuri and we thank him for his help. We also
thank Chandra for several other fruitful discussions and suggestions.
\newpage
\bibliographystyle{plain}
\bibliography{cds}


\end{document}